\documentclass{article}

\usepackage{graphicx}

\usepackage{amsthm}
\usepackage{amsmath}
\usepackage{amssymb}
\usepackage{stmaryrd} 
\usepackage[noblocks]{authblk} 
\usepackage{hyperref} 

\usepackage{tikz}
\usetikzlibrary{automata, positioning}
\usepackage[all]{xy}

\usepackage{bussproofs} 

\newcommand{\MSet}{{M\text{-}\mathbf{Set}}}
\newcommand{\MNom}{{M\text{-}\mathbf{Nom}}}

\newcommand{\CoAlg}{\ensuremath{\mathrm{CoAlg}}}

\newcommand{\permset}{{\perm\text{-}\mathbf{Set}}}
\newcommand{\permnom}{{\perm\text{-}\mathbf{Nom}}}
\newcommand{\sbset}{{\sb\text{-}\mathbf{Set}}}
\newcommand{\sbnom}{{\sb\text{-}\mathbf{Nom}}}

\newcommand{\dom}{\mathrm{dom}}
\DeclareMathOperator{\ev}{ev} 
\DeclareMathOperator{\orb}{orb}

\newcommand{\pow}{\mathcal{P}}
\newcommand{\perm}{\mathsf{Pm}}
\renewcommand{\sb}{\mathsf{Sb}}
\newcommand{\supp}{\mathsf{supp}}
\newcommand{\swap}[2]{(#1 \, #2)}

\newcommand{\atoms}{\mathbb{A}}
\newcommand{\separated}{\mathop{\#}}
\newcommand{\fsto}{\mathrel{\to_\text{\upshape fs}}}
\newcommand{\permto}{\mathrel{\to_\text{\upshape fs}^{\perm}}}
\newcommand{\sbto}{\mathrel{\to_\text{\upshape fs}^{\sb}}}
\newcommand{\mto}{\mathrel{\to_\text{\upshape fs}^{M}}}
\newcommand{\wandto}{\mathrel{-\kern-0.3em{\ast}}}

\newcommand{\sepprod}{\mathop{\ast}}

\DeclareMathOperator{\id}{id}
\newcommand{\sa}{B_*}

\newcommand{\Put}{\mathop{\text{\upshape \sf Put}}}
\newcommand{\Pop}{\text{\upshape \sf Pop}}

\theoremstyle{plain}
\newtheorem{theorem}{Theorem}
\newtheorem{proposition}{Proposition}
\newtheorem{lemma}{Lemma}
\newtheorem{corollary}{Corollary}

\theoremstyle{definition}
\newtheorem{definition}{Definition}

\theoremstyle{remark}
\newtheorem{remark}{Remark}
\newtheorem{example}{Example}

\title{Separation and Renaming in Nominal Sets} 
\author{Joshua Moerman\thanks{This research has been partially funded by the ERC AdG project 787914 FRAPPANT.}}
\affil{RWTH Aachen, Germany, \quad \href{mailto:joshua@cs.rwth-aachen.de}{\texttt{joshua@cs.rwth-aachen.de}}}
\author{Jurriaan Rot\thanks{Also affiliated to Radboud University. The research of this author received funding from the European Union's Horizon 2020 research and innovation programme under the Marie Sk\l{}odowska-Curie Grant Agreement No. 795119.}}
\affil{University College London, UK, \quad \href{mailto:jrot@cs.ru.nl} {\texttt{jrot@cs.ru.nl}}}

\begin{document}
\maketitle
\begin{abstract} 
Nominal sets provide a foundation for reasoning about names.
They are used primarily in syntax with binders, but also, e.g., to model automata over infinite alphabets.
In this paper, nominal sets are related to \emph{nominal renaming sets}, which involve arbitrary substitutions rather than permutations, through a categorical adjunction.
In particular, the left adjoint relates the separated product of nominal sets to the Cartesian product of nominal renaming sets.
Based on these results, we define the new notion of \emph{separated nominal automata}.
These automata can be exponentially smaller than classical nominal automata, if the semantics is closed under substitutions.
\end{abstract}

\section{Introduction}

Nominal sets are abstract sets which allow one to reason over sets with names, in terms
of permutations and symmetries. 
Since their introduction in computer science~\cite{GabbayP99}, they
have been widely used for implementing and reasoning over syntax with binders~\cite{pitts13}. 
Further, nominal techniques have been related to 
computability theory~\cite{BojanczykKLT13} and automata theory~\cite{BojanczykKL14},
where they provide an elegant means of studying languages over infinite alphabets. 
This embeds nominal techniques in a broader setting of \emph{symmetry aware computation}~\cite{Pitts16}.

Gabbay, one of the pioneers of nominal techniques described a variation 
on the theme: \emph{nominal renaming sets}~\cite{Gabbay07,GabbayH08}.
Nominal renaming sets are equipped with a monoid action of arbitrary (possibly non-injective) substitution of names,
in contrast to nominal sets, which only involve a group action of permutations. 

In this paper, we systematically relate nominal renaming sets to nominal sets.
We start by establishing a categorical adjunction (Section~\ref{sec:adjunction}):
\begin{equation*}\label{eq:adj-intro}
\xymatrix{
    \permnom \ar@/^10pt/[rr]^F
    &
    \bot
    &
    \sbnom \ar@/^10pt/[ll]^U
}, 
\end{equation*}
where $\permnom$ is the usual category of nominal sets and $\sbnom$ the category of nominal renaming sets.
The right adjoint $U$ simply forgets the action of non-injective substitutions. The left adjoint $F$
freely extends a nominal set with elements representing the application of such substitutions. 
For instance, $F$ maps the nominal set $\atoms^{(*)}$ of all words consisting of distinct atoms 
to the nominal renaming set $\atoms^*$ consisting of all words over the atoms. 

In fact, the latter follows from one of the main results of this paper:
 $F$ 
maps the \emph{separated product} $X \sepprod Y$ of nominal sets to the Cartesian product of nominal renaming sets. 
Additionally, under certain conditions, $U$ maps the exponent to the \emph{magic wand}
$X \wandto Y$, which is the right adjoint of the separated product. 
The separated product consists of those pairs whose elements have disjoint supports. 
This is relevant for name abstraction~\cite{pitts13},
and has also been studied in the setting of presheaf categories, aimed towards separation logic~\cite{DBLP:journals/jfp/OHearn03}. 

We apply these connections between nominal sets and renaming sets in the context of automata theory. 
Nominal automata are expressively equivalent to the more classical register automata~\cite{Bojanczyk18},
and have appealing properties that register automata lack, such as unique minimal automata. 
However, moving from register automata to nominal automata can lead to an exponential blow-up in the number of states.%
\footnote{Here `number of states' refers to the number of orbits in the state space.} 

As a motivating example,
we consider a language modelling an $n$-bounded FIFO queue.
The input alphabet is given by $\Sigma = \{ \Put(a) \mid a \in \atoms \} \cup \{ \Pop \}$,
and the output alphabet by $O = \atoms \cup \{ \perp \}$ (here, $\perp$ is a \emph{null} value).
The (generalised) language $L_n \colon \Sigma^* \to O$ maps a sequence of queue
operations to the resulting top element when starting from the empty queue, or
to $\perp$ if this is undefined.
The language $L_n$ can be recognised by a nominal (Moore) automaton, but this requires an exponential
number of states in $n$, as the automaton distinguishes internally between all possible equalities
among elements in the queue~\cite{MSSKS17}. 

Based on the observation that $L_n$ is closed under substitutions, 
we can come up with a \emph{linear} automata-theoretic representation.
To this end, we define the new notion of \emph{separated nominal automaton},
where the transition function is only defined for pairs of states and letters with a disjoint support (Section~\ref{sec:automata}).
Using the aforementioned categorical framework,
we can go back and forth between languages from separated automata and languages which are closed under substitutions.
In the FIFO example, the separated automaton obtained from the original nominal automaton has only $n+1$ states,
thus dramatically reducing the number of states.
We expect that such a reduction is useful in many applications, such as automata learning~\cite{MSSKS17}.

\section{Monoid actions and nominal sets}

In order to capture both the standard notion of nominal sets~\cite{pitts13} 
and sets with more general renaming actions~\cite{GabbayH08}, we start by defining monoid actions.

\begin{definition}
Let $(M, \cdot, 1)$ be a monoid. An \emph{$M$-set} is a set $X$ together with a 
function ${\cdot} \colon M \times X \to X$ such that $1 \cdot x = x$ 
and $m \cdot (n \cdot x) = (m \cdot n) \cdot x$ for all $m, n \in M$ and $x \in X$.
The function ${\cdot}$ is called an \emph{$M$-action} and $m \cdot x$ is often written by juxtaposition $m x$.
A function $f \colon X \to Y$ between two $M$-sets is \emph{$M$-equivariant} 
if $m \cdot f(x) = f(m \cdot x)$ for all $m \in M$ and $x \in X$.
The class of $M$-sets together with equivariant maps forms a category $\MSet$.
\end{definition}
Let $\atoms = \{a, b, c, \ldots\}$ 
be a countable infinite set of \emph{atoms}. 
The two main instances of $M$ considered in this paper are the monoid 
$$
\sb = \{m \colon \atoms \to \atoms \mid m(a) \neq a \text{ for finitely many }a\}
$$ 
of all (finite) substitutions (with composition as multiplication), and the monoid
$$
\perm = \{g \in \sb \mid g \text{ is a bijection}\}
$$
of all (finite) permutations. 
Since $\perm$ is a submonoid of $\sb$, any $\sb$-set is also a
$\perm$-set; and any $\sb$-equivariant map is also $\perm$-equivariant. 
This gives rise to a forgetful functor
\begin{equation}\label{eq:forget}
    U \colon \sbset \to \permset .
\end{equation}

The set $\atoms$ is an $\sb$-set by defining $m \cdot a = m(a)$.
Given an $M$-set $X$, the set $\pow(X)$ of subsets of $X$ is an $M$-set,
with the action defined by direct image. 

For a $\perm$-set $X$, the \emph{orbit} of an element $x$ is the set $\orb(x) = \{g \cdot x \mid g \in \perm \}$.
We say $X$ is \emph{orbit-finite} if the set $\{\orb(x) \mid x \in X\}$ is finite.

For any monoid $M$, the category $\MSet$ is symmetric monoidal closed. 
The product of two $M$-sets is given by the Cartesian product, with the 
action defined pointwise: $m \cdot (x, y) = (m \cdot x, m \cdot y)$.
In $\MSet$, the exponent $X \to^{M} Y$ is given by the set
$\{f \colon M \times X \to Y \mid f \text{ is equivariant}\}$.%
\footnote{If we write a regular arrow $\to$, then we mean a map in the category. Exponent objects will always be denoted by annotated arrows.}
The action on such an $f \colon M \times X \to Y$ is defined by
$(m \cdot f)(n,x) = f(mn,x)$.
A good introduction to the construction of the exponent is given by Simmons~\cite{Simmons}.
If $M$ is a group, a simpler description of the exponent may be given, 
carried by the set of all functions $f \colon X \to Y$,
with the action given by $(g \cdot f)(x) = g \cdot f(g^{-1} \cdot x)$.

\subsection{Nominal $M$-sets}

The notion of \emph{nominal} set is usually defined w.r.t.\ a $\perm$-action. Here,
we use the generalisation to $\sb$-actions from~\cite{GabbayH08}.
Throughout this section, let $M$ denote a submonoid of $\sb$. 

\begin{definition}
    Let $X$ be an $M$-set, and $x \in X$ an element. A set
    $C \subset \atoms$ is an \emph{($M$)-support} of $x$ if for all $m_1, m_2 \in M$ s.t.\ $m_1|_C = m_2|_C$ we have $m_1 x = m_2 x$.
    An $M$-set $X$ is called \emph{nominal} if every element $x$ has a finite $M$-support.
\end{definition}

Nominal $M$-sets and equivariant maps form a full subcategory of $\MSet$,
denoted by $\MNom$.
The $M$-set $\atoms$ of atoms is nominal. 
The powerset $\pow(X)$ of a nominal set is not nominal in general; the restriction to
finitely supported elements is. 

If $M$ is a group, then the notion of support can be simplified by using inverses.
To see this, first note that, given elements $g_1, g_2 \in M$, 
$g_1|_C = g_2|_C$ can equivalently be written as $g_1 g_2^{-1}|_C = \id|_C$.
Second, the statement $x g_1 = x g_2$ can be expressed as $x g_1 g_2^{-1} = x$.
Hence, $C$ is a support iff $g|_C = \id_C$ implies $gx = x$ for all $g$,
which is the standard definition for nominal sets over a group~\cite{BojanczykKL14,pitts13}.
Surprisingly, a similar characterisation also holds for $\sb$-sets~\cite{GabbayH08}.
Moreover, recall that every $\sb$-set is also a $\perm$-set; the associated
notions of support coincide on nominal $\sb$-sets, as shown
by the following result.
In particular, this means
that the forgetful functor~\eqref{eq:forget} restricts to $U \colon \sbnom \to \permnom$. 

\begin{lemma}\cite[Theorem~4.8]{Gabbay07}
    \label{lem:GM-support}
    Let $X$ be a nominal $\sb$-set, $x \in X$,
    and $C \subset \atoms$. 
    Then $C$ is an $\sb$-support of $x$ iff it is a $\perm$-support of $x$. 
\end{lemma}

\begin{remark}
It is not true that any $\perm$-support is an $\sb$-support.
The condition that $X$ is nominal, in the above lemma, is crucial.
Let $X = \atoms + 1$ and define the following $\sb$-action: $m \cdot a = m(a)$ if $m$ is injective, $m \cdot a = \ast$ if $m$ is non-injective, and $m \cdot \ast = \ast$.
This is a well-defined $\sb$-set, but is \emph{not nominal}.
Now consider $U(X)$, this is the $\perm$-set $\atoms + 1$ with the natural action, which is a \emph{nominal} $\perm$-set!
In particular, as a $\perm$-set each element has a finite support, but as a $\sb$-set the supports are infinite.

This counterexample is similar to the ``exploding nominal sets'' in \cite{Gabbay07}, but even worse behaved.
We like to call them \emph{nuclear sets}, since an element will collapse when hit by a non-injective map, no matter how far away the non-injectivity occurs.
\end{remark}

For $M \in \{\sb, \perm\}$,
any element $x \in X$ of a nominal $M$-set $X$ has a least finite support (w.r.t.\ set inclusion).
We denote the least finite support of an element $x \in X$ by $\supp(x)$.
Note that by Lemma~\ref{lem:GM-support}, the set $\supp(x)$ is independent
of whether a nominal $\sb$-set $X$ is viewed as an $\sb$-set or a $\perm$-set.
The \emph{dimension} of $X$ is given by $\dim(X) = \max \{|\supp(x)| \mid x \in X \}$,
where $|\supp(x)|$ is the cardinality of $\supp(x)$. 

We list some basic properties of nominal $M$-sets, which have
known counterparts for the case that $M$ is a group~\cite{BojanczykKL14},
and when $M=\sb$~\cite{GabbayH08}.

\begin{lemma}
    \label{lem:transfer-support}
    Let $X$ be an $M$-nominal set.
    If $C$ supports an element $x \in X$,
    then $m \cdot C$ 
    supports $m \cdot x$ for all $m \in M$.
    Moreover, any $g \in \perm$ preserves least supports: $g \cdot \supp(x) = \supp(g x)$.
\end{lemma}

The latter equality does not hold in general for a monoid $M$.
For instance, the `exploding' nominal renaming sets \cite{GabbayH08} give counterexamples for $M = \sb$.

\begin{lemma}
    Given $M$-nominal sets $X, Y$ and a map $f \colon X \to Y$,
    if $f$ is $M$-equivariant and $C$ supports an element $x \in X$,
    then $C$ supports $f(x)$.
\end{lemma}

The category $\MNom$ is symmetric monoidal closed, 
with the product inherited from $\MSet$, thus simply given by Cartesian product.
The exponent is given by the restriction of the exponent $X \to^{M} Y$ in $\MSet$ to the set of finitely
supported functions, denoted by $X \mto Y$.
This is similar to the exponents of nominal sets with 01-substitutions from \cite{Pitts14}.

\begin{remark}
In~\cite{GabbayH08} a different presentation of the exponent in $\MNom$ is given, based on
a certain extension of partial functions. We prefer the previous characterisation,
as it is derived in a straightforward way from the exponent in $\MSet$.
\end{remark}

\subsection{Separated product}

\begin{definition}
Two elements $x, y \in X$ of a $\perm$-nominal set are called \emph{separated}, denoted by $x \separated y$, if there are disjoint sets $C_1, C_2 \subset \atoms$ such that $C_1$ supports $x$ and $C_2$ supports $y$.
The \emph{separated product} of $\perm$-nominal sets $X$ and $Y$ is defined as
$$ X \sepprod Y = \{ (x, y) \mid x \separated y \}. $$
\end{definition}

We extend the separated product to the \emph{separated power}, defined by $X^{(0)} = 1$ and $X^{(n+1)} = X^{(n)} \sepprod X$, and the \emph{set of separated words} $X^{(\ast)} = \bigcup_i X^{(i)}$.
The separated product is an equivariant subset $X \sepprod Y \subseteq X \times Y$.
Consequently, we have equivariant projection maps $X \sepprod Y \to X$ and $X \sepprod Y \to Y$.

\begin{example}
Two finite sets $C, D \subset \atoms$ are separated precisely when they are disjoint.
An important example is the set $\atoms^{(\ast)}$ of separated words over the atoms: it consists of those words where all letters are distinct.
\end{example}

The separated product gives rise to another symmetric closed monoidal structure on $\permnom$, 
with $1$ as unit, and the exponential object given by \emph{magic wand} $X \wandto Y$.
An explicit characterisation of $X \wandto Y$ is not needed in the remainder of this paper,
but for a complete presentation we briefly recall the description from~\cite{Schoepp06} (see also~\cite{pitts13} and~\cite{Clouston13}).
First, define a $\perm$-action on the set of partial functions $f \colon X \rightharpoonup Y$
by $(g \cdot f)(x) =  g \cdot f(g^{-1} \cdot x)$ if $f(g^{-1} \cdot x)$ is defined. 
Now, such a partial function $f \colon X \rightharpoonup Y$ is called \emph{separating}
if $f$ is finitely supported, $f(x)$ is defined iff $f \separated x$, and
$\supp(f) = \bigcup_{x \in \dom(f)} \supp(f(x)) \setminus \supp(x)$. 
Finally, $X \wandto Y = \{f \colon X \rightharpoonup Y \mid f \text{ is separating}\}$. 
See~\cite{Schoepp06} for a proof and explanation.

\begin{remark}\label{rem:atom-abstr}
The special case $\atoms \wandto Y$ coincides with $[\atoms]Y$, the set of \emph{name abstractions}~\cite{pitts13}. 
The latter is generalised to $[X]Y$ in~\cite{gabbay2002fm}. In~\cite{Clouston13} it is shown
that the coincidence $[X]Y \cong (X \wandto Y)$ only holds under strong assumptions
(including that $X$ is single-orbit).
\end{remark}

\begin{remark}
An analogue of the separated product does not seem to exist for nominal $\sb$-sets.
For instance, consider the set $\atoms \times \atoms$.
As a $\perm$-set, it has four equivariant subsets: $\emptyset, \{(a, a) \mid a \in \atoms\}$, $\atoms \sepprod \atoms$, and $\atoms \times \atoms$.
However, the set $\atoms \sepprod \atoms$ is not an equivariant subset when considering $\atoms \times \atoms$ as an $\sb$-set.
\end{remark}

\section{A monoidal construction from $\perm$-sets to $\sb$-sets}\label{sec:adjunction}

In this section, we provide a free construction, extending nominal $\perm$-sets to nominal $\sb$-sets.
We use this as a basis to relate the separated product and exponent (in $\permnom$)
to the product and exponent in $\sbnom$. The main results are:
\begin{enumerate}
\item Theorem~\ref{thm:adjunction}:
the forgetful functor $U \colon \sbnom \to \permnom$ has a left adjoint $F$;

\item Theorem~\ref{thm:monoidal}:
this $F$ is monoidal: it maps separated products to products;

\item Theorem~\ref{thm:exponent-separated} and Corollary~\ref{cor:exponent-iso}:
$U$ maps the exponent object in $\sbnom$ to the right adjoint $\wandto$ of the separated product, if the domain has dimension smaller or equal to $1$.
\end{enumerate}
Together, these results form the categorical infrastructure to relate 
nominal languages to separated languages and automata in Section~\ref{sec:automata}.

\begin{definition}\label{def:left-adjoint}
    Given a $\perm$-nominal set $X$,
    we define a nominal $\sb$-set $F(X)$ as follows.
    Define the set
    \[ F(X) = \{ (m, x) \mid m \in \sb, x \in X \} / _\sim, \]
    where $\sim$ is the least equivalence relation containing:
    \begin{align}
    \label{eq:rel1} (m, g x) &\sim (m g, x), \\
    \label{eq:rel2} (m, x)   &\sim (m', x) \quad \text{ if } m|_C = m'|_C \text{ for a }\perm\text{-support } C \text{ of } x,
    \end{align}
    for all $x \in X$, $m, m' \in \sb$ and $g \in \perm$.
    Note that $mg = m \circ g$, i.e., simply the monoid operation of $\sb$.
    The equivalence class of a pair $(m, x)$ is denoted by $[m, x]$.
    We define an $\sb$-action on $F(X)$ as $n \cdot [m, x] = [n m, x]$.
\end{definition}
Well-definedness is proved as part of Proposition~\ref{prop:functor} below.
Informally, an equivalence class $[m,x] \in F(X)$ behaves ``as if $m$ acted on $x$''.
The first equation~(\ref{eq:rel1}) ensures compatibility with the $\perm$-action on $x$, and the second equation~(\ref{eq:rel2}) ensures that $[m,x]$ only depends the relevant part of $m$.
\begin{example}
Here are a few examples of the application of $F$. We do not give direct proofs, but the first 
two will be treated more systematically later in this section (see Corollary~\ref{prop:An-iso}).
For the third, note that $\atoms \times \atoms$ consist of two orbits, $\atoms \sepprod \atoms$ and the diagonal $\{ (a, a) \mid a \in \atoms \}$.
\begin{itemize}
\item $F(\atoms) \cong \atoms$.
\item $F(\atoms^{(*)}) \cong \atoms^*$.
\item $F(\atoms \times \atoms) \cong \atoms^2 + \atoms$.
\end{itemize}
\end{example}

The following characterisation of $\sim$ is useful in proofs.
(This lemma is proven in the Appendix.)
\begin{lemma}\label{lm:sim}
We have
 $(m_1, x_1) \sim (m_2, x_2)$ iff there 
is a permutation $g \in \perm$ such that $g x_1 = x_2$ and $m_1|_C = m_2 g|_C$, for $C$ some $\perm$-support of $x_1$. 
\end{lemma}

\begin{remark}
The first relation~(\ref{eq:rel1}) in Definition~\ref{def:left-adjoint} comes from the construction of ``extension of scalars'' in commutative algebra \cite{AtiyahM69}.
In that context, one has a ring homomorphism $f \colon A \to B$ and an $A$-module $M$ and wishes to obtain a $B$-module.
This is constructed by the tensor product $B \otimes_A M$ and it is here that the relation $(b, am) \sim (ba, m)$ is used ($B$ is a right $A$-module via $f$).
\end{remark}

\begin{proposition}\label{prop:functor}
    The construction $F$ in Definition~\ref{def:left-adjoint} extends to a functor
    \[ F \colon \permnom \to \sbnom \,, \]
    defined on an equivariant map $f \colon X \to Y$ by $F(f)([m, x]) = [m, f(x)] \in F(Y)$.
\end{proposition}
\begin{proof}
We first prove well-definedness and then the functoriality.

\textbf{$F(X)$ is an $\sb$-set.}
To this end we check that the $\sb$-action is well-defined.
Let $[m_1, x_1] = [m_2, x_2] \in F(X)$ and let $m \in \sb$.
By Lemma~\ref{lm:sim}, there is some permutation $g$ such that  $g x_1 = x_2$
and $m_1|_C = m_2 g|_C$ for some support $C$ of $x_1$.
By post-composition with $m$ we get $m m_1|_C = m m_2 g|_C$, which means (again by the lemma)
that $[m m_1, x_1] = [m m_2, x_2]$.
Thus $m [m_1, x_1] = m [m_2, x_2]$, which concludes well-definedness.

For associativity and unitality of the $\sb$-action, we simply note that it is
directly defined by left multiplication of $\sb$ which is associative and unital.
This concludes that $F(X)$ is an $\sb$-set.

\textbf{$F(X)$ is a nominal $\sb$ set.}
Given an element $[m, x] \in F(X)$ and a $\perm$-support $C$ of $x$, we will prove that $m \cdot C$ 
is an $\sb$-support for $[m, x]$.
Suppose that we have $m_1, m_2 \in \sb$ such that $m_1|_{m\cdot C} = m_2|_{m\cdot C}$.
By pre-composition with $m$ we get $m_1 m|_C = m_2 m|_C$ and this leads us to conclude $[m_1 m, x] = [m_2 m, x]$.
So $m_1 [m,x] = m_2 [m, x]$ as required.

\textbf{Functoriality.}
Let $f \colon X \to Y$ be a $\perm$-equivariant map.
To see that $F(f)$ is well-defined consider $[m_1, x_1] = [m_2, x_2]$.
By Lemma~\ref{lm:sim}, there is a permutation $g$ such that  $g x_1 = x_2$
and $m_1|_C = m_2 g|_C$ for some support $C$ of $x_1$.
Applying $F(f)$ gives on one hand $[m_1, f(x_1)]$ and
on the other hand $[m_2, f(x_2)] = [m_2, f(g x_1)] = [m_2, g f(x_1)] = [m_2 g, f(x_1)]$
(we used equivariance in the second step). Since $m_1|_C = m_2 g|_C$ and $f$ preserves supports
we have $[m_2 g, f(x_1)] = [m_1, f(x_1)]$.

For $\sb$-equivariance we consider both $n \cdot F(f)([m, x]) = n [m, f(x)] = [n m, f(x)]$ and $F(f)(n \cdot [m, x]) = F(f)([nm, x]) = [nm, f(x)]$.
This shows that $n F(f)([m, x]) = F(f)(n [m, x])$ and concludes that we have a map $F(f) \colon F(X) \to F(Y)$.

The fact that $F$ preserves the identity function and composition follows from the definition directly.
\end{proof}

\begin{theorem}
    \label{thm:adjunction}
    The functor $F$ is left adjoint to $U$:
    $$
    \xymatrix{
        \permnom \ar@/^10pt/[rr]^F
        &
        \bot
        &
        \sbnom \ar@/^10pt/[ll]^U
    }
    $$
\end{theorem}
\begin{proof}
    We show that, for every nominal set $X$, there is a map $\eta_X \colon X \to UF(X)$
    with the necessary universal property: for every $\perm$-equivariant $f \colon X \to U(Y)$ there
    is a unique $\sb$-equivariant map $f^\sharp \colon FX \to Y$ such that $U(f^\sharp) \circ \eta_X = f$.
    Define $\eta_X$ by $\eta_X(x) = [\id,x]$.
    This is equivariant: $g \cdot \eta_X(x) = g [\id,x] = [g,x] = [\id, gx] = \eta_X(gx)$.
    Now, for $f \colon X \rightarrow U(Y)$,
    define $f^{\sharp}([m,x]) = m \cdot f(x)$ for $x \in X$ and $m \in \sb$.
    Then $U(f^\sharp) \circ \eta_X(x) = f^\sharp([\id, x]) = \id \cdot f(x) = f(x)$.

    To show that $f^{\sharp}$ is well-defined, consider $[m_1, x_1] = [m_2, x_2]$
    (we have to prove that $m_1 \cdot f(x_1) = m_2 \cdot f(x_2)$).
    By Lemma~\ref{lm:sim}, there is a $g \in \perm$
    such that $g x_1 = x_2$ and $m_2 g |_C = m_1|_C$ for a $\perm$-support $C$ of $x_1$.
    Now $C$ is also a $\perm$-support for $f(x)$ and hence it is an $\sb$-support of $f(x)$ (Lemma~\ref{lem:GM-support}).
    We conclude that $m_2 \cdot f(x_2) = m_2 \cdot f(g x_1) = m_2 g \cdot f(x_1)  = m_1 \cdot f(x_1)$
    (we use $\perm$-equivariance in the one but last step and $\sb$-support in the last step).
    Finally, $\sb$-equivariance of $f^\sharp$ and uniqueness are straightforward calculations (see Appendix).
\end{proof}

The counit $\epsilon \colon FU(Y) \to Y$ is defined by $\epsilon([m, x]) = m \cdot x$.
For the inverse of $-^{\sharp}$, let $g \colon F(X) \to Y$ be an $\sb$-equivariant map; 
then $g^{\flat} \colon X \to U(Y)$ is given by $g^{\flat}(x) = g([\id, x])$.
Note that the unit $\eta$ is a $\perm$-equivariant map, hence it preserves supports (i.e., any support of $x$ also supports $[\id, x]$).
This also means that if $C$ is a support of $x$, then $m \cdot C$ is a support of $[m, x]$ (by Lemma~\ref{lem:transfer-support}).

\subsection{On (separated) products}

The functor $F$ not only preserves coproducts, being a left adjoint, but 
it also maps the separated product to products:

\begin{theorem}
    \label{thm:monoidal}
    The functor $F$ is strong monoidal, from the monoidal category $(\permset, \sepprod, 1)$ to $(\sbset, \times, 1)$.    In particular, the map $p$ given by
    \[ p = \langle F(\pi_1), F(\pi_2) \rangle \colon F(X \sepprod Y) \to F(X) \times F(Y) \]
    is an isomorphism, natural in $X$ and $Y$.
\end{theorem}
\begin{proof}
    We prove that $p$ is an isomorphism.
    It suffices to show that $p$ is injective and surjective (Lemma~\ref{lem:inj-surj-iso}).
    Note that $p([m, (x, y)]) = ([m, x], [m, y])$.

    \textbf{Surjectivity.}
    Let $([m_1, x], [m_2, y])$ be an element of $F(X) \times F(Y)$.
    We take an element $y' \in Y$ such that $y' \separated \supp(x)$ and $y' = gy$ for some $g \in \perm$.
    Now we have an element $(x, y') \in X \sepprod Y$.
    By Lemma~\ref{lem:transfer-support}, we have $\supp(y') = \supp(y)$.
    Define the map
    \[ m(x) = \begin{cases}
        m_1(x)         &\text{if } x \in \supp(x) \\
        m_2(g^{-1}(x)) &\text{if } x \in \supp(y') \\
        x              &\text{otherwise}.
    \end{cases}\]
    (Observe that $\supp(x) \separated \supp(y')$, so the cases are not overlapping.)
    The map $m$ is an element of $\sb$.
    Now consider the element $z = [m, (x, y')] \in F(X \sepprod Y)$.
    Applying $p$ to $z$ gives the element $([m, x], [m, y'])$.
    First, we note that $[m, x] = [m_1, x]$ by the definition of $m$.
    Second, we show that $[m, y'] = [m_2, y]$.
    Observe that $m g|_{\supp(y)} = m_2|_{\supp(y)}$ by definition of $m$.
    Since $\supp(y)$ is a support of $y$, we have $[mg, y] = [m_2, y]$,
    and since $[mg, y] = [m, g y] = [m, y']$ we are done.
    Hence $p([m, (x, y')]) = ([m, x], [m, y']) = ([m_1, x], [m_2, y])$, so $p$ is surjective.

    \textbf{Injectivity.}
    Let $[m_1, (x_1, y_1)]$ and $[m_2, (x_2, y_2)]$ be two elements.
    Suppose that they are mapped to the same element, i.e., $[m_1, x_1] = [m_2, x_2]$ and $[m_1, y_1] = [m_2, y_2]$.
    Then there are permutations $g_x, g_y$ such that $x_2 = g_x x_1$ and $y_2 = g_y y_1$.
    Moreover, let $C=\supp(x_1)$ and $D=\supp(y_1)$; then we have $m_1|_C =  m_2 g_x|_C$ and $m_1|_D =  m_2 g_y|_D$. 
    In order to show the two original elements are equal, we have to provide a single permutation $g$.
    Define for, $z \in C \cup D$,
    \[ g_0(z) = \begin{cases}
        g_x(z) &\text{if } z \in C \\
        g_y(z) &\text{if } z \in D.
    \end{cases} \]
    (Again, $C$ and $D$ are disjoint.)
    The function $g_0$ is injective since the least supports of $x_2$ and $y_2$ are disjoint. Hence $g_0$ defines a local isomorphism from $C \cup D$ to $g_0(C \cup D)$.
    By homogeneity~\cite{pitts13}, the map $g_0$ extends to a permutation $g \in \perm$ with $g(z) = g_x(z)$ for $z \in C$ and $g(z) = g_y(z)$ for $z \in D$.
    In particular we get $(x_2, y_2) = g (x_1, y_1)$.
    We also obtain $m_1|_{C \cup D} = m_2 g|_{C \cup D}$.
    This proves that $[m_1, (x_1, y_1)] = [m_2, (x_2, y_2)]$, and so the map $p$ is injective.

    \textbf{Unit and coherence.}
    To show that $F$ preserves the unit, we note that $[m, 1] = [m', 1]$ for every $m, m' \in \sb$, as the empty set supports $1$ and so $m|_{\emptyset} = m'|_{\emptyset}$ vacuously holds.
    We conclude $F(1)$ is a singleton.
\end{proof}

Since $F$ also preserves coproducts (being a left adjoint), we obtain that 
$F$ maps the set of separated words to the set of all words.

\begin{corollary}\label{cor:sep-words}
    For any $\perm$-nominal set $X$, we have $F(X^{(*)}) \cong (FX)^*$.
\end{corollary}

As we will show below, the functor $F$ preserves
the set $\atoms$ of atoms. This is an instance of a more general result
about preservation of one-dimensional objects.

\begin{proposition}
    \label{lem:1dim-iso}
    The functors $F$ and $U$ are equivalences on $\leq 1$-dimensional objects.
    Concretely, for $X \in \permnom$ and $Y \in \sbnom$:
    \begin{enumerate}
    \item If $\dim(X) \leq 1$, then the unit $\eta \colon X \to UF(X)$ is an isomorphism.
    \item If $\dim(Y) \leq 1$, then the co-unit $\epsilon \colon FU(X) \to X$ is an isomorphism.
    \end{enumerate}
\end{proposition}

Before we prove this proposition, we need the following technical property of $\leq 1$-dimensional $\sb$-sets.

\begin{lemma}\label{lem:1dim-sbnom}
Let $Y$ be a nominal $\sb$-set.
If an element $y \in Y$ is supported by a singleton set (or even the empty set), then
\[ \{ m y \mid m \in \sb \} = \{ g y \mid g \in \perm \} \,. \]
\end{lemma}
\begin{proof}
Let $y \in Y$ be supported by $\{ a \}$ and let $m \in \sb$.
Now consider $b = m(a)$ and the bijection $g = \swap{a}{b}$.
Now $m|_{\{a\}} = g|_{\{a\}}$, meaning that $m y = g y$.
So the set $\{ m y \mid m \in \sb \}$ is contained in $\{ g y \mid g \in \perm \}$.
The inclusion the other way is trivial, which means $\{ m y \mid m \in \sb \} = \{ g y \mid g \in \perm \}$.
\end{proof}

\begin{proof}[Proof of Proposition~\ref{lem:1dim-iso}]
It is easy to see that $\eta \colon x \mapsto [\id, x]$ is injective.
Now to see that $\eta$ is surjective, let $[m, x] \in UF(X)$ and consider a support $\{ a \}$ of $x$ (this is a singleton or empty since $\dim(X) \leq 1$).
Let $b = m(a)$ and consider the swap $g = \swap{a}{b}$.
Now $[m, x] = [m g^{-1}, g x]$ and note that $\{ b \}$ supports $g x$ and $m g^{-1}|_{\{b\}} = \id|_{\{b\}}$.
We continue with $[m g^{-1}, g x] = [\id, g x]$, which concludes that $g x$ is the preimage of $[m, x]$.
Hence $\eta$ is an isomorphism.

To see that $\epsilon \colon [m, y] \mapsto m y$ is surjective, just consider $m = \id$.
To see that $\epsilon$ is injective, let $[m, y], [m', y'] \in FU(Y)$ be two elements such that $m y = m' y'$.
Then by using Lemma~\ref{lem:1dim-sbnom} we find $g, g' \in \perm$ such that $g y = m y = m' y' = g' y'$.
This means that $y$ and $y'$ are in the same orbit (of $U(Y)$) and have the same dimension.
Case 1: $\supp(y) = \supp(y') = \emptyset$, then $[m, y] = [\id, y] = [\id, y'] = [m', y']$.
Case 2: $\supp(y) = \{ a \}$ and $\supp(y') = \{ b \}$, then $\supp(g y) = \{ g(a) \}$ (Lemma~\ref{lem:transfer-support}).
In particular we now now that $m$ and $g$ map $a$ to $c = g(a)$, likewise $m'$ and $g'$ map $b$ to $c$.
Now $[m, y] = [m, g^{-1} g' y'] = [m g^{-1} g', y'] = [m', y']$, where we used $m g^{-1} g (b) = c = m'(b)$ in the last step.
This means that $\epsilon$ is injective and hence an isomorphism.
\end{proof}

By Proposition~\ref{lem:1dim-iso}, we may consider the set $\atoms$ as both $\sb$-set and $\perm$-set (abusing notation).
And we get an isomorphism $F(\atoms) \cong \atoms$ of nominal $\sb$-sets.
To appreciate the above results, we give a concrete characterisation of one-dimensional nominal sets:

\begin{lemma}\label{lm:char-dim-one}
    Let $X$ be a nominal $M$-set, for $M \in \{\sb,\perm\}$.
    Then $\dim(X) \leq 1$ iff there exist (discrete) sets $Y$ and $I$ such that $X \cong Y + \coprod_{I} \atoms$.
\end{lemma}

In particular, the one-dimensional objects include the alphabets used for \emph{data words}, consisting of a product $S \times \atoms$ of a discrete set $S$ of action labels and the set of atoms.
These alphabets are very common in the study of register automata (see, e.g., \cite{IsbernerHS14}).

By the above and Theorem~\ref{thm:monoidal},
$F$ maps separated powers of $\atoms$ to powers,
and the set of separated words over $\atoms$ to the $\sb$-set of words over $\atoms$. 
\begin{corollary}
    \label{prop:An-iso}
    We have $F(\atoms^{(n)}) \cong \atoms^n$ and $F(\atoms^{(*)}) \cong \atoms^{*}$.
\end{corollary}

\subsection{On exponents}

We have described how $F$ and $U$ interact with (separated) products.
In this section, we establish a relationship between the magic wand ($\wandto$) and the exponent of nominal $\sb$-sets ($\sbto$).

\begin{definition}
Let $X \in \permnom$ and $Y \in \sbnom$.
We define a $\perm$-equivariant map $\phi$ as follows:
\[ \phi \colon (X \wandto U(Y)) \to U(F(X) \sbto Y) \]
is defined by using the composition
\[ F(X \wandto U(Y)) \times F(X) \xrightarrow{p^{-1}} F((X \wandto U(Y)) \sepprod X) \xrightarrow{F(\ev)} FU(Y) \xrightarrow{\epsilon} Y, \]
where $p^{-1}$ is from Theorem~\ref{thm:monoidal} and $\ev$ is the evaluation map of the exponent $\wandto$.
By Currying and the adjunction, we obtain $\phi$:

\begin{prooftree}
\AxiomC{$F(X \wandto U(Y)) \times F(X) \to Y$}
\RightLabel{\, by Currying}
\UnaryInfC{$F(X \wandto U(Y)) \to (F(X) \sbto Y)$}
\RightLabel{\, by Theorem~\ref{thm:adjunction}}
\UnaryInfC{$\phi \colon (X \wandto U(Y)) \to U(F(X) \sbto Y)$}
\end{prooftree}
\end{definition}

With this map we can prove a generalisation of Theorem~\ref{thm:adjunction}.
In particular, the following theorem generalises the one-to-one correspondence between maps $X \to U(Y)$ and maps $F(X) \to Y$.
First, it shows that this correspondence is $\perm$-equivariant.
Second, it extends the correspondence to all finitely supported maps and not just the equivariant ones.

\begin{theorem}\label{thm:exponent-separated}
    The sets $X \wandto U(Y)$ and $U(F(X) \sbto Y)$ are naturally isomorphic via $\phi$ as nominal $\perm$-sets.
\end{theorem}
\begin{proof}
We define some additional maps in order to construct the inverse of $\phi$.
First, from Theorem~\ref{thm:adjunction} we get the following isomorphism:
\[ q \colon U(X \times Y) \xrightarrow{=} U(X) \times U(Y) \]
Second, with this map and Currying, we obtain the following two natural maps:

\begin{prooftree}
\AxiomC{$U(F(X) \sbto Y) \times UF(X) \xrightarrow{q^{-1}} U((F(X) \sbto Y) \times F(X)) \xrightarrow{U(\ev)} U(Y)$}
\RightLabel{\, by Currying}
\UnaryInfC{$\alpha \colon U(F(X) \sbto Y) \to (UF(X) \permto U(Y))$}
\end{prooftree}

\begin{prooftree}
\AxiomC{$(UF(X) \permto U(Y)) \times X \xrightarrow{\id {\times} \eta} (UF(X) \permto U(Y)) \times UF(X) \xrightarrow{\ev} U(Y)$}
\RightLabel{\, by Currying}
\UnaryInfC{$\beta \colon (UF(X) \permto U(Y)) \to (X \permto U(Y))$}
\end{prooftree}

Last, we note that the inclusion $A \sepprod B \subseteq A \times B$ induces a \emph{restriction} map $r \colon (B \permto C) \twoheadrightarrow (B \wandto C)$ (again by Currying).
A calculation shows that $r \circ \beta \circ \alpha$ is the inverse of $\phi$ (see Appendix).
\end{proof}

Note that this theorem gives an alternative characterisation of the magic wand in terms of the exponent in $\sbnom$, if the codomain is $U(Y)$.
Moreover, for a $1$-dimensional object $X$ in $\sbnom$, we obtain the following special case of the theorem (using the co-unit isomorphism from Proposition~\ref{lem:1dim-iso}):

\begin{corollary}\label{cor:exponent-iso}
    Let $X, Y$ be nominal $\sb$-sets.
    For $1$-dimensional $X$, the nominal $\perm$-set $U(X) \wandto U(Y)$ is naturally isomorphic to $U(X \sbto Y)$.
\end{corollary}

\begin{remark}
The set $\atoms \wandto U(X)$ coincides with the atom abstraction $[\atoms] UX$ (Remark~\ref{rem:atom-abstr}).
Hence, as a special case of Corollary~\ref{cor:exponent-iso}, we recover~\cite[Theorem 34]{GabbayH08},
which states a bijective correspondence between $[\atoms] UX$ and $U(\atoms \sbto X)$.
\end{remark}

\section{Nominal and separated automata}\label{sec:automata}

In this section, we study nominal (Moore) automata, which recognise languages over infinite alphabets.
After recalling the basic definitions,
we introduce a new variant of automata based on the separating product, 
which we call \emph{separated nominal automata}. 
These automata represent nominal languages which are 
$\sb$-equivariant, essentially meaning they are closed under substitution. 
Our main result is that, if a `classical' nominal automaton (over $\perm$) 
represents a language $L$ which is $\sb$-equivariant, then $L$
can also be represented by a separated nominal automaton. The latter
can be exponentially smaller (in number of orbits) than the original automaton,
as we show in a concrete example.

\begin{remark}
We will work with a general output set $O$ instead of just acceptance.
The reason for this is that $\sb$-equivariant functions $L \colon \atoms \to 2$ are not very interesting: they are defined purely by the length of the input.
By using more general output $O$, we may still capture interesting behaviour, e.g., the language in Example~\ref{ex:fifo}.
\end{remark}

\begin{definition}\label{def:nominal-aut}
Let $\Sigma, O$ be $\perm$-sets, called input/output alphabet respectively. 
\begin{itemize}
\item
A \emph{($\perm$)-nominal language} is an equivariant map of the form $L \colon \Sigma^* \to O$.
\item 
A \emph{nominal (Moore) automaton} $\mathcal{A} = (Q,\delta,o,q_0)$ consists of a nominal set of states $Q$,
an equivariant transition function $\delta \colon Q \times \Sigma \to Q$,
an equivariant output function $o \colon Q \rightarrow O$, and
an initial state $q_0 \in Q$ with an empty support. 

\item 
The \emph{language semantics} is the map $l \colon Q \times \Sigma^* \to O$,
defined inductively by
    $$
    l(x, \varepsilon) = o(x)\,, \qquad l(x, aw) = l(\delta(x,a),w)
    $$
for all $x \in Q$, $a \in \Sigma$ and $w \in \Sigma^*$.
\item For $l^\flat \colon Q \to (\Sigma^* \permto O)$ the transpose of $l$,
we have that $l^\flat(q_0) \colon \Sigma^* \to O$ is equivariant;
this is called the \emph{language accepted by $\mathcal{A}$}.
\end{itemize}
\end{definition}

Note that the language accepted by an automaton can equivalently be characterised by considering paths through the automaton from the initial state.

If the state space $Q$ and the alphabets $\Sigma, O$ are orbit finite, this 
allows us to run algorithms (reachability, minimization, etc.) on such automata \cite{BojanczykKL14}, 
but there is no need to assume this for now. 
For an automaton $\mathcal{A} = (Q,\delta,o,q_0)$, we define the set of 
\emph{reachable states} as the least set $R(\mathcal{A}) \subseteq Q$ such that 
$q_0 \in R(\mathcal{A})$ and for all $x \in R(\mathcal{A})$ and $a \in \Sigma$, $\delta(x,a) \in R(\mathcal{A})$. 

\begin{example}\label{ex:fifo}
    We model a bounded FIFO queue of size $n$ as a nominal Moore automaton,
    explicitly handling the data in the automaton structure.%
    \footnote{We use a reactive version of the queue data structure which is slightly different from the versions in \cite{MSSKS17,IsbernerHS14}.}
    The input alphabet $\Sigma$ and output alphabet $O$ are as follows:
    \[ \Sigma = \{ \Put(a) \mid a \in \atoms \} \cup \{ \Pop \},
    \qquad O = \atoms \cup \{ \perp \}. \]
    
    The input alphabet encodes two actions:
    putting a new value on the queue and popping a value.
    The output is either a value (the front of the queue) or $\perp$ if the queue is empty.
    A queue of size $n$ is modelled by the automaton $(Q, \delta, o, q_0)$ defined as follows.
    \[ Q = \atoms^{\leq n} \cup \{ \perp \}
    \qquad q_0 = \varepsilon
    \qquad o(a_1 \ldots a_k) =
            \begin{cases}
                a_1   &\text{if } k \geq 1 \\
                \perp &\text{otherwise} 
            \end{cases} \]
    \[ \begin{array}{rcl}
        \delta(a_1 \ldots a_k, \Put(b)) &=& 
            \begin{cases}
                a_1 \ldots a_k b & \text{if } k < n \\
                \perp            &\text{otherwise}
            \end{cases}
            \qquad 
            \delta(\bot, x) = \perp \\[10pt]
    \delta(a_1 \ldots a_k, \Pop) &=&
            \begin{cases}
                a_2 \ldots a_k &\text{if } k > 0 \\
                \perp          &\text{otherwise}
            \end{cases} 
        \end{array} \]
    The automaton is depicted in Figure~\ref{fig:fifo} for the case $n = 3$.
    The language accepted by this automaton assigns to a word $w$ the first element
    of the queue after executing the instructions in $w$ from left to right, and
    $\perp$ if the input is ill-behaved, i.e., $\Pop$ is applied to an empty queue or $\Put(a)$ to a full queue.

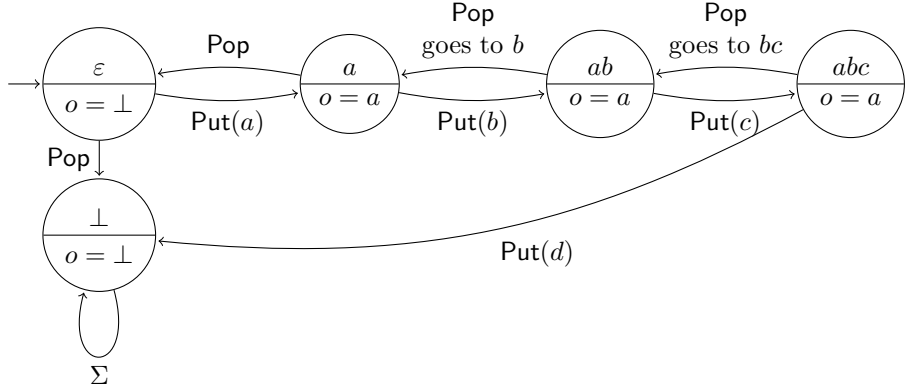
\begin{figure}
\centering
\begin{tikzpicture}[node distance=3.33cm, shorten >=1pt, bend angle=10, initial text=]
  \node[state with output, initial] (0) {$\varepsilon$ \nodepart{lower} $o = {\perp}$};
  \node[state with output, right of=0] (1) {$a$     \nodepart{lower} $o = a$};
  \node[state with output, right of=1] (2) {$a b$   \nodepart{lower} $o = a$};
  \node[state with output, right of=2] (3) {$a b c$ \nodepart{lower} $o = a$};
  \node[state with output, below=.5cm of 0] (s) {$\perp$ \nodepart{lower} $o = {\perp}$};
  
  \path[->]
  (0) edge [bend right] node [below] {$\Put(a)$} (1)
  (0) edge              node [left]  {$\Pop$} (s)
  (1) edge [bend right] node [above] {$\Pop$} (0)
  (1) edge [bend right] node [below] {$\Put(b)$} (2)
  (2) edge [bend right] node [above, align=center] {$\Pop$ \\ goes to $b$} (1)
  (2) edge [bend right] node [below] {$\Put(c)$} (3)
  (3) edge [bend right] node [above, align=center] {$\Pop$ \\ goes to $b c$} (2)
  (3) edge [bend left=17] node [below right] {$\Put(d)$} (s)
  (s) edge [loop below] node {$\Sigma$} (s);
\end{tikzpicture}
\caption{The FIFO automaton from Example~\ref{ex:fifo} with $n = 3$.
The right-most state consists of \emph{five} orbits as we can take $a, b, c$ distinct, all the same, or two of them equal in three different ways.
Consequently, the complete state space has ten orbits.
The output of each state is denoted in the lower part.}
\label{fig:fifo}
\end{figure}
\end{example}

\begin{definition}\label{def:sep-aut}
Let $\Sigma, O$ be $\perm$-sets.
A \emph{separated language} is an equivariant map of the form $\Sigma^{(*)} \to O$.
A \emph{separated automaton} $\mathcal{A} = (Q,\delta,o,q_0)$ 
consists of $Q$, $o$ and $q_0$ defined as in a nominal automaton,
and an equivariant transition function $\delta \colon Q \sepprod \Sigma \rightarrow Q$. 

The \emph{separated language semantics} of such an automaton is given by the map $s \colon Q \sepprod \Sigma^{(*)} \to O$,
defined by
\[ s(x, \varepsilon) = o(x)\,, \qquad s(x, aw) = s(\delta(x,a),w) \]
for all $x \in Q$, $a \in \Sigma$ and $w \in \Sigma^{(*)}$ such that $x \separated aw$ and $a \separated w$.

Let $s^\flat \colon Q \to (\Sigma^{(*)} \wandto O)$ be the transpose of $s$.
Then $s^\flat(q_0) \colon \Sigma^{(*)} \to O$ corresponds to a separated language,
this is called the \emph{separated language accepted by $\mathcal{A}$}.
\end{definition}

By definition of the separated product, the transition function is only defined
on a state $x$ and letter $a \in \Sigma$ if $x \separated a$. 
In Example~\ref{ex:sep-aut-fifo} below, we describe the bounded FIFO 
as a separated automaton, and describe its accepted language.

First, we show how the language semantics of separated nominal automata
extends to a language over \emph{all} words,
provided that both the input alphabet $\Sigma$ and the output alphabet $O$ are $\sb$-sets.

\begin{definition}
    Let $\Sigma$ and $O$ be nominal $\sb$-sets.
    An $\sb$-equivariant function $L \colon \Sigma^* \to O$ is called an \emph{$\sb$-language}.
\end{definition}

Notice the difference between an $\sb$-language $L \colon \Sigma^* \to O$ 
and a $\perm$-language $L' \colon (U \Sigma)^* \to U(O)$. They are both functions from
$\Sigma^{*}$ to $O$, but the latter is only $\perm$-equivariant,
while the former satisfies the stronger property of $\sb$-equivariance. 
Languages over separated words, and $\sb$-languages, are connected as follows.

\begin{proposition}\label{thm:extension}
    Suppose $\Sigma, O$ are both nominal $\sb$-sets, and suppose $\dim(\Sigma) \leq 1$. 
    There is a one-to-one correspondence 
    $$
    \frac{S \colon (U\Sigma)^{(*)} \rightarrow UO \quad\text{ $\perm$-equivariant}}
    {\overline{S} \colon \Sigma^* \rightarrow O \quad \text{ $\sb$-equivariant}}
    $$
    between separated languages and $\sb$-nominal languages. 
    From  $\overline{S}$ to $S$, this is given by application of the forgetful functor and restricting to the subset of separated words.
    
    For the converse direction, given $w = a_1 \ldots a_n \in \Sigma^{*}$, let
    $b_1, \ldots, b_n \in \Sigma$ such that
        $w \separated b_i$ for all $i$, and
    $b_i \separated b_j$ for all $i,j$ with $i \neq j$.
    Define $m \in \sb$ by 
    $$m(a) = \begin{cases} a_i & \text{ if } a = b_i \text{ for some }i \\ a & \text{ otherwise} \end{cases}$$
    Then 
    $\overline{S}(a_1a_2a_3\cdots a_n) = m \cdot S(b_1 b_2 b_3 \cdots b_n)$.
\end{proposition}
\begin{proof}
There is the following chain of one-to-one correspondences, from the results of the previous section:
\begin{prooftree}
\AxiomC{$(U\Sigma)^{(*)} \rightarrow UO$}
\RightLabel{\, by Theorem~\ref{thm:adjunction}}
\UnaryInfC{$F(U\Sigma)^{(*)} \rightarrow O$}
\RightLabel{\, by Corollary~\ref{cor:sep-words}}
\UnaryInfC{$(FU\Sigma)^* \rightarrow O$}
\RightLabel{\, by Proposition~\ref{lem:1dim-iso}}
\UnaryInfC{$\Sigma^* \rightarrow O$}
\end{prooftree}
\end{proof}

Thus, every separated automaton over $U(\Sigma), U(O)$ gives rise to an $\sb$-language $\overline{S}$,
corresponding to the language $S$ accepted by the automaton.

Any nominal automaton $\mathcal{A}$ restricts to a separated automaton, formally described in Definition~\ref{def:restr-aut}.
It turns out that if the $(\perm)$-language accepted by $\mathcal{A}$ is actually an $\sb$-language, 
then the restricted automaton already represents this language, as the extension $\overline{S}$
of the associated separated language $S$ (Proposition~\ref{thm:separated-sb-lang}). Hence, in such a case,
the restricted separated automaton suffices to describe the language of $\mathcal{A}$.

\begin{definition}\label{def:restr-aut}
Let $i \colon Q \sepprod U(\Sigma) \hookrightarrow Q \times U(\Sigma)$ be the natural inclusion map.
A nominal automaton $\mathcal{A} = (Q, \delta, o, q_0)$ 
induces a separated automaton $\mathcal{A}_*$, by setting $\mathcal{A}_* = (Q, \delta \circ i, o, q_0)$.
\end{definition}

\begin{proposition}\label{thm:separated-sb-lang}
    Suppose $\Sigma, O$ are both $\sb$-sets, and suppose $\dim(\Sigma) \leq 1$.
    Let $L \colon (U\Sigma)^* \to UO$ be the $\perm$-nominal language accepted by a nominal automaton $\mathcal{A}$,
    and suppose $L$ is $\sb$-equivariant.
    Let $S$ be the separated language accepted by $\mathcal{A}_*$.
    Then $L = U(\overline{S})$.
\end{proposition}
\begin{proof}
It follows from the one-to-one correspondence in Proposition~\ref{thm:extension}:
on the bottom there are two languages ($L$ and $U(\overline{S})$), while there is only the restriction of $L$ on the top.
We conclude that $L = U(\overline{S})$.
\end{proof}

As we will see in Example~\ref{ex:sep-aut-fifo}, separated automata allow us to represent $\sb$-languages 
in a much smaller way than nominal automata. Given a nominal automaton $\mathcal{A}$, a smaller
separated automaton can be obtained by computing the reachable part of the restriction $\mathcal{A}_*$.
The reachable part is defined similarly (but only where $\delta$ is defined) and denoted by $R(\mathcal{A}_*)$ as well.

\begin{lemma}
For any nominal automaton $\mathcal{A}$, we have
$R(\mathcal{A}_*) \subseteq R(\mathcal{A})$.
\end{lemma}

The converse inclusion of the above proposition does certainly not hold, as shown by the following example.

\begin{example}\label{ex:sep-aut-fifo}
    Let $\mathcal{A}$ be the automaton modelling a bounded FIFO queue (for some $n$), from Example~\ref{ex:fifo}.
    The $\perm$-nominal language $L$ accepted by $\mathcal{A}$ is $\sb$-equivariant:
    it is closed under application of arbitrary substitutions.
    
    The separated automaton $\mathcal{A}_*$
    is given simply by restricting the transition function
    to $Q \sepprod \Sigma$, i.e., a $\Put(a)$-transition from a state $w \in Q$
    exists only if $a$ does not occur in $w$.
    The separated language $S$ accepted by this new automaton is the restriction
    of the nominal language of $\mathcal{A}$ to separated words.
    By Proposition~\ref{thm:separated-sb-lang}, we have $L = U(\overline{S})$.
    Hence, the separated automaton $\mathcal{A}_*$ represents $L$,
    essentially by closing the associated separated language $S$ under all substitutions.
    
    The \emph{reachable} part of $\mathcal{A}_*$ is given by
    \[ R_{\mathcal{A}_*} =  \atoms^{(\leq n)} \cup \{ \perp \} \,. \]
    Clearly, restricting $\mathcal{A}_*$ to the reachable part does not affect the accepted language.
    However, while the orginal state space $Q$ has exponentially many orbits in $n$,
    $R_{\mathcal{A}_*}$ has only $n+1$ orbits!
    Thus, taking the reachable part of $R_{\mathcal{A}_*}$ yields a separated automaton which
    represents the FIFO language $L$ in a much smaller way than the original automaton.
\end{example}

\subsection{Separated automata: coalgebraic perspective}
\label{sec:sep-aut}

Nominal automata and separated automata can be presented as \emph{coalgebras}
on the category of $\perm$-nominal sets. In this section we revisit
the above results from this perspective, and generalise from 
(equivariant) languages to finitely supported languages. In particular,
we retrieve the extension from separated languages to $\sb$-languages,
by establishing $\sb$-languages as a final separated automaton. The latter
result follows by instantiating a well-known technique for lifting adjunctions
to categories of coalgebras, using the results of Section~\ref{sec:adjunction}.
In the remainder of this section we assume familiarity with the theory
of coalgebras, see, e.g.,~\cite{jacobs-coalg,Rutten00}.

\begin{definition}\label{def:nominal-aut-coalg}
Let $M$ be a submonoid of $\sb$, and let $\Sigma$, $O$
be nominal $M$-sets, referred to as the input and output alphabet respectively. We define the functor
$B_M \colon \MNom \to \MNom$ by $B_M(X) = O \times (\Sigma \mto X)$.
An \emph{($M$)-nominal (Moore) automaton} is a $B_M$-coalgebra.
\end{definition}

A $B_M$-coalgebra can be presented as a nominal set $Q$ together with the pairing
\[ \langle o, \delta^\flat \rangle \colon Q \to O \times (\Sigma \mto Q) \]
of an equivariant \emph{output} function $o \colon Q \to O$,
and (the transpose of) an equivariant \emph{transition} function $\delta \colon Q \times \Sigma \to Q$.
In case $M = \perm$, this coincides with the
automata of Definition~\ref{def:nominal-aut}, omitting initial states. 
The language semantics is generalised accordingly, as follows.
Given such a $B_M$-coalgebra $(Q, \langle o, \delta^\flat \rangle)$, 
the \emph{language semantics} $l \colon Q \times \Sigma^* \to O$
is given by
\begin{equation}\label{eq:lang-sem-coalg}
l(x, \varepsilon) = o(x)\,, \qquad l(x, aw) = l(\delta(x,a),w)
\end{equation}
for all $x \in S$, $a \in \Sigma$ and $w \in \Sigma^*$. 

\begin{proposition}\label{thm:final-nom-aut}
    Let $M$ be a submonoid of $\sb$, let $\Sigma$, $O$
be nominal $M$-sets. 
    The nominal $M$-set $\Sigma^* \mto O$ 
    extends to a final $B_M$-coalgebra $(\Sigma^* \mto O, \zeta)$,
    such that the unique homomorphism from a given $B_M$-coalgebra
    is the transpose $l^\flat$ of the language semantics~\eqref{eq:lang-sem-coalg}.
\end{proposition}

A \emph{separated automaton} (Definition~\ref{def:sep-aut}, without initial states)
corresponds to a coalgebra for the functor $\sa \colon \permnom \to \permnom$ 
given by $\sa(X) = O \times (\Sigma \wandto X)$. 
The separated language semantics arises by finality.

\begin{proposition}\label{thm:final-sep}
    The set $\Sigma^{(\ast)} \wandto O$ is the carrier of a final $\sa$-coalgebra,
    such that the unique coalgebra homomorphism from a given
    $\sa$-coalgebra $(Q,\langle o, \delta \rangle)$
    is the transpose $s^\flat\!$
    of the separated language semantics $s \colon Q \sepprod \Sigma^{(\ast)} \to O$
    (Definition~\ref{def:sep-aut}). 
\end{proposition}

Next, we provide an alternative description of the
final $\sa$-coalgebra which assigns $\sb$-nominal languages
to states of separated nominal automata. 
The essence is to obtain
a final $\sa$-coalgebra from the final $B_{\sb}$-coalgebra. 
In order to prove this, we use a technique to lift adjunctions to categories of coalgebras.
This technique occurs regularly in the coalgebraic study
of automata~\cite{JSS14,KlinR16,KerstanKW14}.

\begin{theorem}\label{thm:adjunction-lift}
    Let $\Sigma$ be a $\perm$-set, and $O$ an $\sb$-set. 
    Define $\sa$ and $B_\sb$ accordingly, as $\sa(X) = UO \times (\Sigma \wandto X)$
    and $B_\sb(X) = O \times (F\Sigma \sbto X)$. 
    There is an adjunction $\overline{F} \dashv \overline{U}$ in:
    $$
    \xymatrix@R=1.5cm{
        \CoAlg(\sa) \ar@/^10pt/[rr]^-{\overline F} 
        & 
        \bot
        &
        \CoAlg(B_\sb) \ar@/^10pt/[ll]^-{\overline U} 
    }
    $$
    where $\overline{F}$ and $\overline{U}$ coincide with $F$ and $U$ respectively on carriers.
\end{theorem}
\begin{proof}
There is a natural isomorphism
$ \lambda \colon \sa U \Rightarrow U B_\sb $
given by
\[ \lambda \colon UO \times (\Sigma \wandto UX) \xrightarrow{\id \times \phi} UO \times U(F \Sigma \sbto X) \xrightarrow{\cong} U(O \times (F \Sigma \sbto X)) \,, \]
where $\phi$ is the isomorphism from Theorem~\ref{thm:exponent-separated} and the isomorphism on the right comes from $U$ being a right adjoint.
The result now follows from Theorem~2.14 in~\cite{HermidaJ98}.
In particular, $\overline{U}(X, \gamma) = (UX, \lambda^{-1} \circ U(\gamma))$.
\end{proof}

Since right adjoints preserve limits, and final objects in particular, we
obtain the following, giving semantics of separated automata through finality.

\begin{corollary}\label{cor:final-coalgs}
    Let $((F \Sigma)^* \sbto O, \zeta)$  be the final $B_{\sb}$-coalgebra (Proposition~\ref{thm:final-nom-aut}). 
    Then the $\sa$-coalgebra $\overline{U}(\Sigma^* \sbto O, \zeta)$ is final and carried
    by the set $ (F \Sigma)^* \sbto O $ of $\sb$-nominal languages. 
\end{corollary}

\section{Related and future work}

Fiore and Turi described a similar adjunction between certain presheaf categories~\cite{FioreT01}.
However, Staton describes in his thesis that the usage of presheaves allows for many degenerate models and one should look at sheaves instead~\cite{Staton07}.
The category of sheaves is equivalent to the category of nominal sets.
Staton transfers the adjunction of Fiore and Turi to the sheaf categories.
We conjecture that the adjunction presented in this paper is equivalent, but defined in more elementary means.
The monoidal property of $F$, which is crucial for our application in automata, has not been discussed before.

An interesting line of research is the generalisation to other symmetries by Boja\'{n}czyk et al.~\cite{BojanczykKL14}.
In particular, the total order symmetry is relevant,
since it allows one to compare elements on their order, as often used in data words.
In this case the symmetries are given by the group of all monotone bijections.
Many results of nominal sets generalise to this symmetry.
For monotone substitutions, however, the situation seems more subtle.
For example, we note that a substitution which maps two values to the same value actually maps \emph{all} the values in between to that value.
Whether the adjunction from Theorem~\ref{thm:adjunction} generalises to other symmetries is left as future work.

This research was motivated by learning nominal automata.
If we know a nominal automaton recognises an $\sb$-language, then we are better off learning a separated automaton directly.
From the $\sb$-semantics of separated automata, it follows that we have a Myhill-Nerode theorem, which means that learning is feasible.
We expect that this can be useful, since we can achieve an exponential reduction this way.

Boja\'{n}czyk et al. prove that nominal automata are equivalent to register automata in terms of expressiveness \cite{BojanczykKL14}.
However, when translating from register automata with $n$ states to nominal automata, we may get exponentially many orbits.
This happens for instance in the FIFO automaton (Example~\ref{ex:fifo}).
We have shown that the exponential blow-up is avoidable by using separated automata, for this example and in general for $\sb$-equivariant languages.
Such languages come from register automata which manipulate data but where do control flow does not depend on comparisons.
This typically occurs in data structures.

An important open problem is whether the latter requirement can be relaxed, by adding separated transitions only locally in a nominal automaton.
A possible step in this direction is to consider the monad $T = UF$ on $\permnom$ and incorporate it in the automaton model.
We believe that this is the hypothesised ``substitution monad'' from~\cite{MSSKS17}.
The monad is monoidal (sending separated products to Cartesian products) and if $X$ is an orbit-finite nominal set, then so is $T(X)$.
This means that we can consider nominal $T$-automata and we can perhaps determinise them using coalgebraic methods~\cite{SilvaBBR13}.

\section*{Acknowledgements}

We would like to thank Gerco van Heerdt and Tom Hirschowitz for their useful comments.

\bibliographystyle{plain}
\bibliography{refs}

\clearpage\appendix
\section{Omitted Proofs}

\begin{lemma}(Lemma~\ref{lem:GM-support})
	Let $X$ be a nominal $\sb$-set and $x \in X$.
	Then a set $C$ supports $x$ w.r.t. the $\sb$-action if and only if $C$ supports $x$ w.r.t. the (induced) $\perm$-action.
\end{lemma}
\begin{proof}
(Slightly different proof from the one in \cite{Gabbay07}.)
First note that one direction is easy: an $\sb$-support is also a $\perm$-support.

For the other direction, let $C$ be a $\perm$-support of $x \in X$.
Now let $C'$ be a $\sb$-support of $x$ (note that $X$ is a nominal $\sb$-set).
Assume $m, m'$ such that $m|_C = m'|_C$ and take a $g$ such that $g|_C = \id|_C$ and $g(c') \in \ker(m, m')$ for all $c' \in C' \setminus C$.
This is possible, since $C' \setminus C$ is a finite set, and there are only finitely many values where $m$ and $m'$ disagree ($\sb$ consists of finite substitutions).
Now $mg|_{C'} = m'g|_{C'}$, so $mg x = m'g x$ (remember that $C'$ is an $\sb$-support).
Using that $C$ is a $\perm$-support we know that $gx = x$ and so $mg x = m'gx \implies mx = m'x$.
This concludes that $C$ is a $\sb$-support.
\end{proof}

\begin{lemma}(Lemma~\ref{lm:sim})
We have
 $(m_1, x_1) \sim (m_2, x_2)$ iff there 
is a $g \in \perm$ such that $g x_1 = x_2$ and $m_1|_C = m_2 g|_C$, for $C$ some $\perm$-support of $x_1$. 
\end{lemma}
\begin{proof}
Let us define $(m_1, x_1) \sim' (m_2, x_2)$ if there exists a permutation $g \in \perm$ such that $g x_1 = x_2$ and $m_1|_C = m_2 g|_C$, for $C$ some $\perm$-support of $x_1$.
We need to prove ${\sim} = {\sim'}$.

(${\sim'} \subseteq {\sim}$)
Let $(m_1, x_1) \sim' (m_2, x_2)$, then there is a $g$ with the required properties.
Now check $(m_2, x_2) = (m_2, g x_1) \sim (m_2 g, x_1) \sim (m_1, x_1)$.

(${\sim} \subseteq {\sim'}$)
For this direction, we show that $\sim'$ is an equivalence relation which contains relations~(\ref{eq:rel1}) and~(\ref{eq:rel2}).
Proving reflexivity and symmetry is trivial.
For transitivity, assume $(m_1, x_1) \sim' (m_2, x_2)$ and $(m_2, x_2) \sim' (m_3, x_3)$.
There are $g, h$ with $g x_1 = x_2$ and $h x_2 = x_3$.
Note that $h g x_1 = x_3$.
Moreover, we have $m_1|_{\supp(x_1)} = m_2 g|_{\supp(x_1)}$ and $m_2|_{\supp(x_2)} = m_3 h|_{\supp(x_2)}$.
Observe that $\supp(x_2) = g \supp(x_1)$ so that $m_2|_{g \supp(x_1)} = m_3 h|_{g \supp(x_1)}$.
Now it follows that $m_1|_{\supp(x_1)} = m_2 g|_{\supp(x_1)} = m_2 h g|_{\supp(x_1)}$.
This proves $(m_1, x_1) \sim' (m_3, x_3)$ as required.
To see that relation~\ref{eq:rel1} is included, we take the $g$ from the relation and check the requirements.
To see that relation~\ref{eq:rel2} is included, we simply take $g = \id$ and the requirement on supports follow by assumption.
\end{proof}

\begin{theorem}(Theorem~\ref{thm:adjunction})
	The functor $F$ is left adjoint to $U$:
	$$
	\xymatrix{
		\permnom \ar@/^10pt/[rr]^F
		&
		\bot
		&
		\sbnom \ar@/^10pt/[ll]^U
	}
	$$
\end{theorem}
\begin{proof}
(Remainder of proof: equivariance and uniqueness)
	For the $\sb$-equivariance, we compute
	$n \cdot f^{\sharp}([m, x]) = n m \cdot f(x) = f^{\sharp}([nm, x]) = f^{\sharp}(n [m, x] )$.
	
	Finally, for uniqueness, suppose $h \colon FX \rightarrow Y$ is such that $U(h) \circ \eta_X = f$, i.e.,
	$h([\id,x]) = f(x)$. Then 
	$
		h([m,x]) = h(m[\id, x]) = m \cdot h([\id,x])  = m \cdot f(x) = m \cdot f^\sharp([\id,x]) = f^\sharp(m \cdot [\id,x]) = f^\sharp([m,x])
	$ as desired.
\end{proof}

\begin{lemma}\label{lem:inj-surj-iso}
	If a map in $\sbnom$ is both injective and surjective, then it is an isomorphism of nominal $\sb$-sets.
\end{lemma}
\begin{proof}
	Let $f \colon X \to Y$ be an injective, surjective map between nominal $\sb$-sets.
	Since $f$ is surjective, it is epic, and
	since it is injective, it is monic.
	Every map which is epic and monic is an isomorphism in $\sbset$, because it is a topos.
	Since $\sbnom$ is a full subcategory of $\sbset$, we conclude that $f$ is an isomorphism of nominal $\sb$-sets.
	(Alternatively, we can use the fact that $\sbnom$ is a topos~\cite{GabbayH08}.)
\end{proof}

\begin{theorem}(Theorem~\ref{thm:exponent-separated})
	The sets $X \wandto U(Y)$ and $U(F(X) \fsto Y)$ are naturally isomorphic via $\phi$ as nominal $\perm$-sets.
\end{theorem}
\begin{proof}
	(Remainder of the proof: calculations of composition)
	We prove that $r \circ \beta \circ \alpha$ is the inverse of $\phi$.%
	\footnote{We use $r_{X, U(Y)} \colon (X \permto U(Y)) \twoheadrightarrow (X \wandto U(Y))$ here.}
	The easiest way to see this is to evaluate $(r \circ \beta \circ \alpha \circ \phi)^{\sharp}$ on a separated pair $(\psi, x)$ and see that this equals $\id^{\sharp}(\psi, x)$.
	(The $\sharp$ denotes (un)currying or transposition.)
	Most steps involve spelling out the definitions.
	We comment only on the interesting steps.
	
	Let $(\psi, x) \in (X \wandto U(Y)) \sepprod X$ and calculate:
	\begin{align*}
	& (r \circ \beta \circ \alpha \circ \phi)^{\sharp} (\psi, x) \\
	 & = r(\beta(\alpha(\phi(\psi))))(x) && = (\ev \circ i)(\beta(\alpha(\phi(\psi))), x) \\
	 & = \beta(\alpha(\phi(\psi)))(x)    && = (\ev \circ \id \times \eta)(\alpha(\phi(\psi)), x)   \\
	 & = \ev(\alpha(\phi(\psi)), [\id, x]) && = \alpha(\phi(\psi))([\id, x]) \\
	 & = (U(\ev) \circ q^{-1}) (\phi(\psi), [\id, x]) && = U(\ev) (\phi(\psi), [\id, x]) \\
	 & =^{1} \phi(\psi)(\id, [\id, x])  && = (\epsilon \circ F(\ev) \circ p^{-1})_{\flat}([\id, \psi])(\id, [\id, x]) \\
	 & =^{2} (\epsilon \circ F(\ev) \circ p^{-1})(\id \cdot [\id, \psi], [\id, x]) && = \epsilon(F(\ev)(p^{-1}([\id, \psi], [\id, x]))) \\
	 & =^{3} \epsilon(F(\ev)([\id, (\psi, x)]) && = \epsilon([\id, \psi(x)]) \\
	 & = \id \cdot \psi(x) && = \psi(x)
	\end{align*}
	
	So $(r \circ \beta \circ \alpha \circ \phi)^{\sharp} = \ev = \id^{\sharp}$.
	Hence $r \circ \beta \circ \alpha \circ \phi = \id$.
	There are few steps which require attention.
	Ad 1: Recall that the exponent in $\sbnom$ consists of two-variable functions and that $\ev(\psi, y) = \psi(\id, y)$.
	Ad 2: Currying in $\sbnom$ is defined by $g_{\flat} c (s, a) = g(s \cdot c, a)$.
	Ad 3: Note that $\psi \separated x$, and so to compute $p^{-1}$ we do not have to rename things to be fresh.
	
	Similarly, we prove that the composition $\phi \circ r \circ \beta \circ \alpha = \id$, by considering its (twice) transpose.
	Let $(\psi, x) \in FU(F(X) \sbto Y) \times F(X)$ and note that $\psi = [m, \psi_0], x = [n, x_0]$.
	Now calculate:
	\begin{align*}
	& (\phi \circ r \circ \beta \circ \alpha)^{\sharp\sharp} (\psi, x) \\
	 & = (\phi \circ r \circ \beta \circ \alpha)^{\sharp} (\psi) (\id, x) && = m \cdot (\phi \circ r \circ \beta \circ \alpha (\psi_0)) (\id, x) \\
	 & = (\phi \circ r \circ \beta \circ \alpha (\psi_0)) (m, x) && = \phi(r(\beta(\alpha(\psi_0)))) (m, x) \\
	 & = (\epsilon \circ F(\ev) \circ p^{-1})_{\flat} ([\id, r(\beta(\alpha(\psi_0)))]) (m, x) && = (\epsilon \circ F(\ev) \circ p^{-1}) (m \cdot [\id, r(\beta(\alpha(\psi_0)))], x) \\
	 & = (\epsilon \circ F(\ev) \circ p^{-1}) ([m, r(\beta(\alpha(\psi_0)))], [n, x_0]) && =^{1} \epsilon(F(\ev)(m', [(r(\beta(\alpha(\psi_0))), x_0')])) \\
	 & = \epsilon([m', r(\beta(\alpha(\psi_0))) (x_0')]) && = m' \cdot r(\beta(\alpha(\psi_0))) (x_0') \\
	 & =^{2} m' \cdot \beta(\alpha(\psi_0)) (x_0') && = m' \cdot (\ev \circ \id \times \eta) (\alpha(\psi_0), x_0') \\
	 & = m' \cdot \ev(\alpha(\psi), [\id, x_0']) && = m' \cdot \alpha(\psi_0) ([\id, x_0']) \\
	 & = m' \cdot \psi_0(\id, [\id, x_0']) && =^{3} \psi_0(m', [m', x_0']) \\
	 & =^{1} \psi_0(m', [n, x_0']) && = (m' \cdot \psi_0) (\id, x) \\
	 & = \epsilon(\psi)(\id, x) && = ({\ev} \circ \epsilon \times \id)(\psi, x)
	\end{align*}
	
	We conclude that $(\phi \circ r \circ \beta \circ \alpha)^{\sharp\sharp} = {\ev} \circ \epsilon \times \id = \id^{\sharp\sharp}$ and so $\phi \circ r \circ \beta \circ \alpha = \id$.
	Ad 1 (twice): To apply $p^{-1}$, we choose a fresh $x_0'$ and define $m'$ such that $[m, \psi_0] = [m', \psi_0]$ and $[n, x_0] = [m', x_0']$.
	Ad 2: Note that $x_0 \separated \psi_0$ (and so $x_0 \separated \beta(\alpha(\psi_0))$), so we are allowed to restrict with $r$.
	Ad 3: Here we use that $\psi_0$ is a (two-variable) $\sb$-equivariant map.
	
	This concludes that $r \circ \beta \circ \alpha = \phi^{-1}$.
\end{proof}

\end{document}